\documentclass[10pt,twocolumn,aps,pra,showpacs,superscriptaddress,floatfix,nofootinbib]{revtex4-1}

\usepackage{amsthm}
\usepackage{amsmath}
\usepackage{latexsym}
\usepackage{amsfonts}
\usepackage{amssymb}
\usepackage{color}
\usepackage{bbm,dsfont}
\usepackage{graphicx}
\usepackage{hyperref}
\usepackage{subfigure}

\newcommand{\id}{\mathbbm{1}} 
\newcommand{\sch}{\mathsf T}

\newcommand{\A}{\mathsf A}

\newtheorem{theorem}{Theorem}
\newtheorem{corollary}{Corollary}
\newtheorem{lemma}{Lemma}
\DeclareMathOperator{\tr}{tr}
\newcommand{\ket}[1]{| #1 \rangle}

\newcommand{\ketbra}[1]{| #1 \rangle\langle #1 |}
\newcommand{\ad}{\Lambda}

\definecolor{green}{RGB}{7,133,10}

\begin{document}
\title{Continuous variable steering and incompatibility via state-channel duality}
\author{Jukka Kiukas}
\email{jek20@aber.ac.uk}
\affiliation{Department of Mathematics, Aberystwyth University, Aberystwyth SY23 3BZ, UK}
\author{Costantino Budroni}
\email{costantino.budroni@oeaw.ac.at}
\affiliation{Institute for Quantum Optics and Quantum Information (IQOQI), Boltzmanngasse 3 1090 Vienna, Austria}
\author{Roope Uola}
\email{roope.uola@gmail.com}
\affiliation{Naturwissenschaftlich-Technische Fakult\"at, Universit\"at 
 Siegen, Walter-Flex-Str.\ 3, D-57068 Siegen, Germany}
\author{Juha-Pekka Pellonp\"a\"a}
\email{juhpello@utu.fi}
\affiliation{Turku Centre for Quantum Physics, Department of Physics and Astronomy, University of Turku, FI-20014 Turku, Finland }

\begin{abstract}  The term Einstein-Podolsky-Rosen steering refers to a quantum correlation intermediate between entanglement and Bell nonlocality, which has been connected to another fundamental quantum property: measurement incompatibility. In the finite-dimensional case, efficient computational methods to quantify steerability have been developed. In the infinite-dimensional case, however, less theoretical tools are available. Here, we approach the problem of steerability in the continuous variable case via a notion of state-channel correspondence, which generalizes the well-known Choi-Jamio\l{}kowski correspondence. Via our approach we are able to generalize the connection between steering and incompatibility to the continuous variable case and to connect the steerability of a state with the incompatibility breaking property of a quantum channel, with applications to noisy NOON states and amplitude damping channels. Moreover, we apply our methods to the Gaussian steering setting, proving, among other things, that canonical quadratures are sufficient for steering Gaussian states.

\end{abstract}

\maketitle

\section{Introduction}
The phenomenon of Einstein-Podolsky-Rosen (EPR) steering combines two central features of quantum 
theory: entanglement and incompatibility, namely, the impossibility of determine precisely and simultaneously certain 
properties of a physical system, e.g., position and momentum. 
In practice, steering is a quantum effect by which one experimenter, Alice, can remotely prepare (i.e., steer) an ensemble 
of states for another experimenter, Bob, by performing local measurement on her half of a bipartite system shared by them, 
and communicating the results to Bob \cite{WiJoDo07}.

Due to the fact that steering is a form of quantum correlation intermediate between entanglement and Bell nonlocality 
\cite{Bell64}, it has been proven  useful to solve foundational problems \cite{Pu13,Moro14,VeBr14,HiQu16,CaGu16} and 
important for applications in quantum information processing such as one-sided-device-independent (1SDI) quantum information 
\cite{BrCa12,KoXi17,XiKo17}. 

In the finite-dimensional case, several methods are available to attack the steering problem. In particular,  efficient 
methods based on semidefinite programming \cite{Boydbook} are able to detect and quantify steerability of a given state and 
set of measurements \cite{Pu13,SW_skr, wp, CaSk16b}. Notwithstanding the existence of several methods, see, e.g., \cite{WiJoDo07, Reid89, Cav09, Schnee13, KoLe15,KoSk15} and the review \cite{CaSk16b}, such a systematic approach  is missing in the continuous variable case. 

In this paper, we will develop a general tool for discussing steering in the continuous variable case, which is based on an 
extension of the Choi-Jamio\l{}kowski state-channel duality \cite{Choi75,Jam72,Ho11}. The Choi-Jamio\l{}kowski 
correspondence associates a state to each channel, but not all states can be mapped to a channel in this way. We will extend 
this idea by showing that one can associate to each bipartite state a channel, such that the steerability property of a 
state is equivalent to the property of the corresponding channel being incompatibility breaking \cite{ibc}, when all possible measurements are allowed for steering. 
This result, in turn, extends to the continuous variable case the result on equivalence between steering and joint-measurability \cite{UoMoGu14,QuVeBr14,UoBu15}. 

In addition to these conceptual results, we find that the channel picture reduces seemingly 
different steering problems to a single one. For instance, we show that steerability of noisy NOON-states (cf.\ Ref.~\cite{KoSk15}) corresponds to the decoherence of incompatibility under an amplitude damping channel (cf.\ \cite{AdHe15, BuKi16}), and how to use steering to investigate its Markovianity properties. Using incompatibility techniques
we investigate both analytically and numerically the noise tolerance of these states with two quadrature measurements. Finally, we apply our methods in the continuous variable Gaussian settings, showing that steerability by a pair of canonical quadrature measurements already ensures steerability by all Gaussian measurements, and connecting this to Gaussian incompatibility breaking channels \cite{gibc}. We also show in passing how the method yields an independent proof of the known Gaussian steering criterion \cite{WiJoDo07}.

The paper is organised as follows: We begin by introducing preliminary notions in Section \ref{sec:prel}, including the general formalism for measurements, joint measurability, steering, the formal connection between hidden state models and positive-operator valued measures (POVMs), and quantification of steering and incompatibility. Section \ref{sec:main} contains our main results on the role of state-channel duality in the connection between steering and incompatibility. In section \ref{sec:appl} we present all the above mentioned applications, except for the Gaussian case, which is treated separately in Section \ref{sec:gauss}. Technical proofs of four Lemmas are given in Appendices \ref{Aps:sc-dual}-\ref{Aps:gauss-dual}, and Appendix \ref{App:LHS} contains the derivation of the Gaussian LHS, which is not essential for understanding the main results.

\section{Preliminary notions}\label{sec:prel}

\subsection{Measurements as POVMs}

A POVM with a discrete outcome set $\Lambda$ is a collection $\{G_\lambda\}_{\lambda\in \Lambda}$ of positive semidefinite operators such that $\sum_{\lambda\in \Lambda} G_\lambda =\id$. Such operators 
represent the probability of the outcome $\lambda$ for a measurement on a state $\rho$ via the rule ${\rm Prob}(\lambda)=\tr [\rho M_\lambda]$. This notion is not sufficient for this paper, since we also consider Gaussian measurements. A POVM with a \emph{continuous outcome set} is one for which $\Lambda=\mathbb R^n$, i.e.\ the Euclidean space. This space comes with the usual integration measure $d\lambda$, and a POVM $\{G_\lambda\}_{\lambda\in \Lambda}$ consist of elements $G_\lambda$ which may be ``infinitesimal'' so that, in general, only the integrals $\int_{X} G_\lambda d\lambda$ with $X\subset \Lambda$ define proper operators. This definition clarifies the name  positive-operator valued measure \cite{BuLaMi96}, i.e. a map from measurable sets to positive operators $X\mapsto \int_X G_\lambda d\lambda$ with normalisation $\int_\Lambda G_\lambda d\lambda=\id$ and countable additivity on disjoint sets. In order to illustrate this well-known technical issue with a typical example relevant for the main text, consider the position operator
$Q = \int_{\mathbb R}q |q\rangle\langle q| dq$. The corresponding POVM has elements $|q\rangle\langle q|$, which are \emph{not} proper operators as they map wave functions $\psi$ into improper states $\psi(q) |q\rangle$. The symbols $|q\rangle\langle q|$ only make up operators when integrated into $\int_{[a,b]} |q\rangle\langle q| dq$, which projects $\psi$ into the wave function coinciding with $\psi(q)$ for $a\leq q\leq b$ and vanishing elsewhere.

\subsection{Joint measurability}

A collection of POVMs, indexed by measurement settings $x$, will be denoted as $\mathcal{M}=\{M_{a|x}\}_{a,x}$ and called a {\it measurement assemblage}. In the discrete case, it is said to be  {\it jointly measurable} \cite{BuLaMi96} 
if there is a POVM $\{G_\lambda\}_\lambda$ such that each $M_{a|x}$ can be obtained from $G_\lambda$ via classical postprocessing, i.e., $M_{a|x} = \sum_\lambda D(a|x,\lambda) G_\lambda$ for all $x,a$,
where $D(a|x,\lambda)\geq 0$ and $\sum_a D(a|x,\lambda)=1$. For the continuous case, with $\mathcal A_x$ the set of outcomes for the POVM $M_x$, one has joint measurability if 
\begin{equation}\label{eq:JMc}
M_{X|x} := \int_{X} M_{a|x} da = \int_\Lambda D(X|x,\lambda) G_\lambda d\lambda,
\end{equation}
where the postprocessing $D(\cdot |x,\cdot)= \mathcal{A}_x\times \Lambda \rightarrow [0,1]$ is generally known as a \emph{weak Markov kernel} \cite{Je08}. An assemblage not jointly measurable is called \emph{incompatible}.

\subsection{Quantum steering}

Another main ingredient for our discussions is bipartite quantum steering. Alice can prepare an ensemble of states for Bob by performing a local measurement ($x$) on her half of the bipartite state $\rho$ and communicating the result ($a$) to Bob. This is related to the measurement assemblage $\{A_{a|x}\}_{a,x}$ via ${\varrho(a|x) := \tr_A[(A_{a|x}\otimes \openone) \rho]/P(a|x)}$, where ${P(a|x):=\tr[(A_{a|x}\otimes \openone) \rho]}$ is the probability of the outcome $a$ for the  setting $x$, and $\varrho(a|x)$ is the reduced state obtained by Bob in this case. We call the collection $\{\rho_{a|x}\}_{a,x}$, with $\rho_{a|x}:=  \tr_\text{A}[(A_{a|x}\otimes \openone) \rho]$, a {\it state assemblage}. It satisfies the nonsignalling rule $\rho_{\mbox{\tiny B}} = \sum_a \rho_{a|x}$ for all $x$, with $\rho_{\mbox{\tiny B}}:=\tr_A[\rho]$ the reduced state for Bob. An assemblage $\{\rho_{a|x}\}_{a,x}$ is called {\it unsteerable} if it admits a local hidden state (LHS) 
model~\cite{WiJoDo07}, i.e., a collection of positive operators  $\{ \sigma_\lambda\}_\lambda$ with 
$\tr[\sum_\lambda \sigma_\lambda]=1$ and $\rho_{a|x} = \sum_\lambda D(a|x,\lambda)\ \sigma_\lambda$ for all  $a,x,$
where $D(a|x,\lambda)\geq 0$ and $\sum_a D(a|x,\lambda) =1$. If a LHS model exists, Bob can 
interpret each $\rho_{a|x}$ as coming from some preexisting states $\sigma_\lambda$, where only the classical probabilities 
are updated due to the information obtained by Alice from her measurement. In the continuous case the assemblage consists of operators $\sigma_x(X):= \int_{X} \sigma_{a|x} da$ where $X\subset \mathcal{A}_x$, and the unsteerable case with LHS $\{ \sigma_\lambda\}_\lambda$ is defined by
\begin{equation}
\int_{X} \sigma_{a|x} da = \int_\Lambda D(X|x,\lambda) \sigma_\lambda d\lambda,
\end{equation}
where $D(\cdot|x,\cdot)$ is a weak Markov kernel for each $x$. In the steerable case we also say that the state $\rho$ \emph{is steerable by} the the measurement assemblage $\{A_{a|x}\}_{a,x}$.

Our main results (Th.~\ref{th:asmb} and \ref{th:gau_st} below) can be applied to reduce seemingly different steering problems to a single one. In order to formulate this precisely, we need a few extra notions. First, we say that states $\rho_1$ and $\rho_2$ are \emph{steering-equivalent} if they are steerable by the exact same measurement assemblages $\{A_{a|x}\}_{a,x}$. For a weaker version, suppose instead that there is a quantum channel $\Lambda$ (with Heisenberg picture $\Lambda^*$) such that $\rho_1$ is steerable by an assemblage $\{A_{a|x}\}_{a,x}$ exactly when $\rho_2$ is steerable by $\{\Lambda^*(A_{a|x})\}_{a,x}$. Generalising the notion in \cite{UoBu15} we, then, call $B_{a|x}:=\Lambda^*(A_{a|x})$ the \emph{steering-equivalent observables} (for $A_{a|x}$). A related (state-independent) notion is that of an {\it incompatibility breaking channel} (IBC) \cite{ibc}, namely, a channel $\Lambda$ such that $\{\Lambda^*(A_{a|x})\}_{x,a}$ is jointly measurable for any measurement assemblage $\{A_{a|x}\}_{x,a}$. For instance, entanglement breaking channels \cite{EBC} belong to this class. It is known \cite{ibc} that when such a channel is applied to one side of a maximally entangled state, the resulting state is not steerable by any measurement assemblage. Cor.~\ref{p:stch} (d) extends this to arbitrary states in the broader context of state-channel duality (see below).

\subsection{Hidden state models and measurements in terms of POVMs}\label{subsec:prel_hidden}
We now review the fact that hidden state models and general quantum observables can both be described by POVMs. Since we are interested in the infinite-dimensional case with POVMs having continuous outcome sets, some technical considerations are unavoidable, and we discuss them briefly. These technicalities are not essential for understanding the main text, but they are needed to make the proofs mathematically sound.

The connection between hidden state models and POVMs is fairly obvious when $d<\infty$ and $\Lambda$ is discrete. Suppose now we have a general family $\{\sigma_\lambda\}_{\lambda\in \Lambda}$ of positive operators on Bob's side of a bipartite setting. Here $\Lambda$ is the set of hidden variables, either discrete or continuous as above. \emph{The crucial difference to POVMs is that each $\sigma_\lambda$ is a proper trace class operator, i.e.\ not ``infinitesimal'' even in the continuous case.} The function $\lambda \mapsto \sigma_\lambda$ must satisfy the technical condition of measurability in the trace class norm, to allow the (Bochner) integrals
$\int f(\lambda) \sigma_\lambda d\lambda$ to exist with finite trace for every measurable scalar function $f$ on $\Lambda$. We also assume the normalisation $\sum_\lambda \sigma_\lambda =\sigma$ (discrete case) and $\int \sigma_\lambda d\lambda =\sigma$ (continuous case), where $\sigma$ is again a fixed density operator. Then there exists a unique POVM $G$ with outcomes in $\Lambda$, satisfying
\begin{equation}\label{grep}
\sigma^{\frac 12}G_\lambda\sigma^{\frac 12} = \sigma_\lambda.
\end{equation}
This is clear in the finite-dimensional case with finite outcome set $\Lambda$ --- we just multiply with $\sigma^{-\frac 12}$ which preserves positivity, and normalisation translates into $\sum_\lambda G_\lambda = \id$. For $d=\infty$ we  need a technical density argument analogous to that used in the proof of Lemma~1 below (see Appendix \ref{Aps:sc-dual}). In the case of continuous outcome set, \eqref{grep} is again understood via the corresponding integrals. 

Suppose then that we start with a POVM $\{G_\lambda\}_{\lambda\in \Lambda}$; the question is how to get the states 
$\sigma_\lambda$. If $\Lambda$ is discrete, this is trivial: we define $\sigma_{\lambda}:=\sigma^{\frac 12}G_\lambda\sigma^{\frac 12}$. However, the case of continuous outcome set $\Lambda$ introduces a subtlety: we have to show that the possibly infinitesimal POVM elements $G_\lambda$ yield trace class operators $\sigma_\lambda$. In general, this is nontrivial, and follows from the Radon-Nikodym property of the trace class (cf.\ p.\ 79 of \cite{diestel}). In the relevant case of a position operator (and more generally a Gaussian POVM), this is easier to prove: $\sigma^{\frac 12}|q\rangle\langle q|\sigma^{\frac 12}$ maps $\psi$ into $\langle q|\sigma^{\frac 12}\psi\rangle \sigma^{\frac 12}|q\rangle$, which is indeed a proper wave function since $\sigma^{\frac 12}|q\rangle = \sum_n \sqrt{s_n} \langle n|q\rangle |n\rangle$ has finite norm $\sum_n s_n |\langle n|q\rangle|^2<\infty$ for all $q$ due to $\sum_n s_n <\infty$, assuming the basis functions are continuous (which is the case for the number basis considered in the main text). 
Here, $\sigma = \sum_n s_n \ketbra{n}$ is the eigendecomposition of $\sigma$.

\subsection{Robustness quantification}

Both incompatibility and steering can be quantified by the amount of classical noise required to destroy these quantum properties. There are different ways of setting up a precise definition for this idea; here we only introduce the quantifiers which turn out to be naturally compatible with our state-channel duality.

We recall from \cite{CavCSR} that \emph{Consistent Steering Robustness} ${\rm CSR}$ of a state assemblage is given by 
\begin{align}
\nonumber {\rm CSR} (\{\sigma_{a|x}\})= \inf &\Big\{t\geq 0 \,\Big|\, \{\pi_{a|x}\}\,\text{ $\sigma$-consistent},\\
& \left\{\frac{\sigma_{a|x} + t \pi_{a|x}}{1+ t}\right\}\,\text{ unsteerable } \Big\},
\end{align}
where $\sigma$-consistence means $\sum_a \sigma_{a|x} = \sum_a \tau_{a|x}$ for all $x$. Similarly, the \emph{Incompatibility Robustness} ${\rm IR}$ \cite{UoBu15} of a measurement assemblage is given by 
\begin{align}\label{eq:IR}
\nonumber {\rm IR}(\{M_{a|x}\}) = \inf \Big\{t\geq 0 \,\Big|\, & \frac{M_{a|x} + t N_{a|x}}{1+ t} \\ &\text{ jointly measurable } \Big\}.
\end{align}
We stress that these definitions, although typically interpreted as SDPs in the finite-dimensional case, can also be stated in infinite dimensions with possibly continuous outcomes for the measurements. We note that in such a case they can only be formulated as SDPs by first restricting to a subspace and discretising the outcomes, as in our numerical example in Section \ref{subsec:noon}.

\section{Main result: state-channel correspondence and steering}\label{sec:main}

Our key idea for attacking steering problems is a state-channel duality valid in infinite dimensions. It goes beyond the familiar Choi-Jamio\l{}kowski (CJ) isomorphism, which maps channels $\sch:\mathcal{L}(\mathcal{H}_B)\rightarrow \mathcal{L}(\mathcal{H}_A)$ into states ${\rho = (\sch\otimes {\rm Id})(\ketbra{\Omega_0})}$ on $\mathcal{H}_A\otimes \mathcal{H}_B$, where $\ket{\Omega_0}=\frac{1}{\sqrt{d}}\sum_n \ket{nn}$ is the maximally entangled state on $\mathcal{H}_B\otimes\mathcal{H}_B$ and ${\rm dim}\mathcal{H}_B=d<\infty$. The CJ isomorphism is a one-to-one map between channels and states $\rho$ with completely mixed $\mathcal{H}_B$ marginals, i.e. $\sigma=\tr_A [\rho]=\openone/d$. It has been used in the definition of channel steering \cite{PianiCS} and the verification of the quantumness of a channel \cite{Pu15}. Our extension is as follows.
\begin{lemma}\label{l:channel}
There is a 1-to-1 correspondence between bipartite states $\rho$ sharing a full-rank marginal $\sigma={\rm tr}_A[\rho]$, and quantum channels $\sch$ from Bob to Alice, such that
\begin{equation}\label{eq:st_ch}
\rho = (\sch \otimes {\rm Id})(|\Omega_\sigma\rangle\langle \Omega_\sigma|)
\end{equation}
where ${|\Omega_\sigma\rangle  := \sum_{n=1}^{d} \sqrt{s_n} |nn\rangle}\in \mathcal H_B\otimes \mathcal H_B$ is defined as the purification of $\sigma = \sum_n s_n \ketbra{n}$.
\end{lemma}
We postpone the detailed proof to Appendix \ref{Aps:sc-dual}. However, since one aim of the paper is to pay due attention to the technicalities related to the infinite-dimensional case, we briefly sketch the relevant points here in the main text: Given a channel $\sch$, $\rho$ is clearly a valid state with $\tr_A[\rho]=\sigma$. Viceversa, given $\rho$ with marginal $\sigma$, the idea is to find a channel $\sch$ such that
\begin{equation}\label{eq:ch_st}
\sigma^{\frac{1}{2}} \sch^*(A) \sigma^{\frac{1}{2}}= \tr_A[\rho (A\otimes \openone)]^{\intercal},
\end{equation}
where the transpose is taken w.r.t. the basis $\{|n\rangle\}$. Eq. \eqref{eq:ch_st} can then be seen to be equivalent to \eqref{eq:st_ch} by direct computation. In order to find $\sch$, one can invert $\sigma^{\frac{1}{2}}$ and solve for $\sch^*(A)$ provided that $d < \infty$. For $d=\infty$, one cannot directly invert $\sigma^{\frac{1}{2}}$, since it will be an unbounded operator. However, one can still construct the Kraus operators $\{M_k\}_k$ for the channel $\sch^*$ from the Kraus operators $R_k$ of $\sigma^{\frac{1}{2}} \sch^*(\cdot) \sigma^{\frac{1}{2}}$, obtained via Eq.~\eqref{eq:ch_st}. This is achieved by extending $R_k\sigma^{-\frac{1}{2}}$ to a bounded operator on $\mathcal{H}_B$; see Appendix \ref{Aps:sc-dual}.

Using the Lemma, we can prove the equivalence between steering of a state assemblage and incompatibility of a measurement assemblage \cite{UoBu15} in full generality and from a quantitative perspective \cite{ChBu16,CavCSR}:
\begin{theorem}\label{th:asmb} The state assemblage $\{\sigma_x(X)\}_{X,x}$ defined by $\rho$ and $\{A_x\}_x$ is steerable $\Leftrightarrow$ the measurement assemblage $\{\sch^*(A_x)\}_x$ is incompatible. Here $\sch\leftrightarrow \rho$ via Lemma \ref{l:channel}, with $\sigma = {\rm tr}_A[\rho]=\sigma_x(\mathcal A_x)$. 
This correspondence is quantitative in that the incompatibility robustness (IR) of $\{\sch^*(A_x)\}_x$ coincides with the consistent steering robustness (CSR) of $\{\sigma_x(X)\}_{X,x}$.
\end{theorem}
\begin{proof}
Using Lemma \ref{l:channel} with any fixed state $\sigma$, we have the correspondences
\begin{align}
\nonumber \{\sch^* (A_{a|x})\}& \mapsto \{ \rho_{a|x} \}, \\
\sch &\mapsto \rho = (\sch \otimes {\rm Id})(|\Omega_\sigma\rangle\langle \Omega_\sigma|),
\end{align}
between the measurement assemblage $A_{a|x}$ transformed, via the Heisenberg-picture channel $\sch^*$ and the steering assemblage obtained via measurements $A_{a|x}$ on the state $\rho$. Note that the measurements $\{ A_{a|x}\}$ stay fixed. Now, $\{\sch^* (A_{a|x})\}$ is jointly measurable if and only if
\begin{equation}
\sch^*(A_{a|x}) = \sum_\lambda D(a|x, \lambda) G_\lambda.
\end{equation}
By multiplying this with $\sigma^{\frac 12}$ on both sides, we obtain
\begin{equation}
\rho_{a|x}^\intercal = \sum_\lambda D(a|x, \lambda) \sigma_\lambda,
\end{equation}
where the hidden states $\sigma_\lambda$ correspond to $G_\lambda$ via \eqref{grep}, and $\rho_{a|x}^\intercal:=\sigma^{\frac 12}\sch^*(A_{a|x})\sigma^{\frac 12}= {\rm tr}_A[\rho(A_{a|x}\otimes \id)]^\intercal$ is the assemblage. As we have established above, all the correspondences are one-to-one, and hence steerability of the setting $(\rho,\{\A_{a|x}\})$ is equivalent to the incompatibility of $\{\sch^* (A_{a|x})\}$.

In order to prove the equivalence of the quantifiers, we follow a similar reasoning as the one in Ref.~\cite{ChBu16}:  We need to prove that for each noise term $N_{a|x}$ of the IR problem, i.e., a term making the measurement assemblage jointly measurable for a given $t$, we can find a noise term $\pi_{a|x}$ of the CSR problem, i.e., a term making the state assemblage unsteerable for the same $t$, and viceversa. We use again the relation
\begin{equation}\label{eq:noise_t}
 \pi_{a|x}^\sch   = \sigma^{\frac 12} N_{a|x} \sigma^{\frac 12}
\end{equation}
to obtain a a one-to-one mapping between $\sigma$-consistent assemblages and arbitrary POVMs. In the finite-dimensional case, we can argue as follows: Given   a $\sigma$-consistent assemblage $\{ \pi_{a|x}\}_{a,x}$,  $\{ N_{a|x}\}_{a,x}$ defined as in Eq.~\eqref{eq:noise_t} is a valid measurement assemblage. Viceversa, given $\{ N_{a|x}\}_{a,x}$ a valid measurement assemblage, we can construct the $\sigma$-consistent assemblage $\{ \pi_{a|x}\}_{a,x}$ as
\begin{equation}
\pi_{a|x}^\sch  = \tr_A\left[ N_{a|x}\otimes \id \ketbra{\Omega_\sigma}\right] = \sigma^{\frac 12} N_{a|x} \sigma^{\frac 12},
\end{equation}
where $\ket{\Omega_\sigma} := \sum_n \sqrt{s_n} \ket{nn}$ is the purification of $\sigma := \sum_n s_n \ketbra{n}$. Hence 
${\rm CSR} (\{\sigma_{a|x}\}) = {\rm IR} (\{\sch^*(A_{a|x})\})$. When the Hilbert space is infinite-dimensional, with possibly continuous outcomes for the POVMs, we again need the same argument as in Section \ref{subsec:prel_hidden}, since $N_{a|x}$ may not be a proper operator, while we need $\pi_{a|x}$ to actually be in the trace class. This establishes the correspondence \eqref{eq:noise_t} between POVMs and $\sigma$-consistent assemblages in the same way as we obtained \eqref{grep}. Then the equality ${\rm CSR} (\{\sigma_{a|x}\}) = {\rm IR} (\{\sch^*(A_{a|x})\})$ clearly follows, and so we can extend the equivalence of quantifiers to the infinite-dimensional case.
\end{proof}

We remark that the above reasoning also provides the connection with the steering equivalent observables defined in the introduction. Given a state assemblage $\{\rho_{a|x}\}_{a,x}$, with a full rank reduced state $\sigma := \sum_a \rho_{a|x}$, its steering equivalent (SE) observables \cite{UoBu15} are given by
\begin{equation}
B_{a|x} := \sigma^{-1/2} \rho_{a|x} \sigma^{-1/2} =\sch^*_\rho(A_{a|x}).
\end{equation}
We stress that we have here used Thm.~\ref{th:asmb} above to make a connection between the notion in \cite{UoBu15} and the one given in the introduction in terms of channels. In particular, this extends the former notion to the infinite-dimensional case.

Furthermore, it is easy to show that if we have only access to the assemblage $\{\rho_{a|x}\}_{a,x}$, and not to the bipartite state $\rho$, we can always interpret $B_{a|x}$ as the observables giving the assemblage when measured on the purification $\ket{\Omega_\sigma} := \sum_n \sqrt{s_n} \ket{nn}$ of $\sigma := \sum_n s_n \ketbra{n}$. Namely,
\begin{equation}
\sigma^{1/2} B_{a|x} \sigma^{1/2} = \tr_A\left[ B_{a|x} \otimes \id \ketbra{\Omega_\sigma}\right]^\sch.
\end{equation}
From Thm.~\ref{th:asmb} we know that $\{\rho_{a|x}\}_{a,x}$ is unsteerable $\Leftrightarrow$ $\{B_{a|x}\}_{a,x}$ is jointly measurable. If we compare that with the definition of the channel $\sch_\rho$, we find that
$\sch^*_{_{\ketbra{\Omega_\sigma}}} (B_{a|x}) = B_{a|x}$. Hence, the observables $B_{a|x}=\sch^*(A_{a|x})$, when measured on $\ket{\Omega_\sigma}$, reproduce the state assemblage $\{\sigma_{a|x}\}_{a,x}$.
We record this conclusion, along with some other direct implications of Thm.~\ref{th:asmb}, into the following Corollary, which generalises several existing results.
\begin{corollary}\label{p:stch} (a) Two states $\rho_1,\rho_2$ are steering-equivalent if the corresponding channels of Lemma~\ref{l:channel} have $\sch_1^*(\cdot)=U\sch_2^*(\cdot)U^*$ where $U$ is unitary. (b) A pure state $\ket{\Psi}$ of full Schmidt rank is steerable by assemblage $\{A_{X|x}\}_{X,x}$ iff the latter is incompatible. (c) A state $\rho$ is steerable by measurements $\{A_{X|x}\}$ iff the purification $\ket{\Omega_\sigma}$ of Lemma~\ref{l:channel} is steerable by the steering-equivalent measurements $\{\sch^*(A_{X|x})\}$. (d) A state $\rho$ is unsteerable iff the channel $\sch^*$ is incompatibility breaking.
\end{corollary}
\begin{proof}
Part (a) follows directly from Thm.~\ref{th:asmb} and the fact that incompatibility is preserved in unitary operations. We demonstrate the use of (a) with NOON-states below. Part (b) is the infinite-dimensional version of the result in Refs.~\cite{UoMoGu14, QuVeBr14}, and can be obtained by defining a Hilbert-Schmidt operator $R$ with $\langle n|R|m\rangle = \langle nm|\Psi\rangle$, where the basis on Bob's side is chosen as in Lemma \ref{l:channel}. Since $R$ and $R^*$ have full rank, $U = R\sigma^{-\frac 12}$ is unitary and 
$\ket{\Psi} = (U\otimes\id)|\Omega_\sigma\rangle$, so that $\sch^*(A) =U^*AU$ and hence preserves incompatibility. Part (c) was proved above, while (d) is a direct consequence of Thm.~\ref{th:asmb} on the theory of incompatibility breaking channels.\end{proof}

We stress the difference with respect to Ref.~\cite{ibc}, where the incompatibility breaking property of a given quantum channel was related to the unsteerability property of specific bipartite states derived from it. Here we have devised a way (via the above state-channel duality) to do the converse: for any given state $\rho$ we can find a quantum channel $\sch$ which is incompatibility breaking exactly when the state is steerable. This allows us to treat \emph{any} given steering problem as an IBC problem, thereby opening up new possibilities for investigating steering. In the following section we illustrate this with different applications.

\section{Applications}\label{sec:appl}

\subsection{Separable and pure states} Consider first separable states $\rho = \sum_i p_i \rho_A^i\otimes \rho_B^i$, which are of course not steerable. We easily find the channel of Lemma~\ref{l:channel} as
$\sch^*(A) = \sum_i {\rm tr}[\rho_A^iA] F_i$, where $F_i = p_i\sigma^{-\frac 12}(\rho_B^i)^\intercal\sigma^{-\frac 12}$ satisfies $0\leq F_i \leq \id$ and $\sum_i F_i =\id$, that is, $\sch$ is entanglement breaking \cite{EBC}.

At the other extreme, pure states of full Schmidt rank correspond to unitary channels by Cor.~\ref{p:stch} (b). As an infinite-dimensional example, the channel for the two-mode coherent state $|z\rangle$ with $z=re^{i\theta}$ is the phase shift $\sch^*(A) = e^{i\theta a^\dagger a}Ae^{-i\theta a^\dagger a}$ if we identify the photon number bases of Alice and Bob. Importantly, the problem of non-unique regularisation of maximally entangled states in $d=\infty$ is circumvented by our method.
 
\subsection{Noisy NOON-states}\label{subsec:noon} Consider the ``NOON-state" $${|N00N\rangle=\frac {1}{\sqrt 2}(|0N\rangle-e^{iN \alpha}|N0\rangle)}$$ shared by Alice and Bob \cite{Le02},  with $\{ \ket{n} \}$ photon number basis of 1-mode electromagnetic field. Via random photon loss, the state becomes
$${\rho_\eta = \eta |N00N\rangle\langle N00N| +(1-\eta) |00\rangle\langle 00|},$$ which is unsteerable for $\eta=0$ and steerable for $\eta=1$. Hence, there is a threshold $\eta_c$ (depending on the allowed measurements) such that $\rho_\eta$ is steerable iff $\eta > \eta_c$ (cf.\ \cite{KoSk15} and the references therein, for previous results on the problem). Using Lemma~\ref{l:channel} we find the channel of $\rho_\eta$ as
\begin{eqnarray}
\nonumber \sch^*(A) &=& \sigma^{-\frac{1}{2}}\tr_A[\rho (A\otimes \openone)]^{\intercal}\sigma^{-\frac{1}{2}} \\
\nonumber &=&
 \begin{pmatrix}
r^2 A_{NN}+(1-r^2)A_{00} & -rA_{N0}e^{-iN\alpha} \\
 -rA_{0N}e^{+iN\alpha} & A_{00} \\
 \end{pmatrix} \\
 &=& U^*\ad_r^*(A)U \label{eq:AD}
\end{eqnarray}
where $r = \sqrt{\eta/(2-\eta)}$, $\sigma={\rm tr}_A(\rho)=(1-\eta/2)|0\rangle\langle 0|+\eta/2|N\rangle\langle N|$,
\begin{equation}\label{eq:damp_matrix}
\ad_r^*(A) = \sum_{i=0}^1 K_{i,r}^*AK_{i,r}=\begin{pmatrix}
A_{00} & rA_{0N} \\
rA_{N0} & r^2 A_{NN}+(1-r^2)A_{00} \\
 \end{pmatrix}
\end{equation}
is the \emph{amplitude damping channel} \cite{NH} with Kraus operators
\begin{equation}\label{eq:damp_kraus}
K_{0,r} = \begin{pmatrix}1 & 0\\ 0 & r\end{pmatrix},  \;  K_{1,r} = \begin{pmatrix}0 & \sqrt{1-r^2}\\ 0 & 0\end{pmatrix},
\end{equation}
and
\begin{equation}
U := |0\rangle\langle N| -e^{iN\alpha} |N\rangle\langle 0| = \begin{pmatrix}
0 & 1 \\
-e^{iN\alpha} & 0\\
\end{pmatrix} 
\end{equation}
is a unitary matrix. By Cor.~\ref{p:stch} (a) , the unitary is irrelevant for steering, and we will ignore it in what follows.

The problem, then, reduces to the question of how $\ad_r$ breaks incompatibility. We introduce the following necessary criterion for this:
\begin{lemma}\label{l:jmad}
Let $\{A_x\}_{x=1}^n$ be any finite assemblage of qubit measurements (with arbitrary outcome sets $\mathcal A_x$). Then the ``damped measurements'' $\ad^*_r(A_x)$ are jointly measurable if 
\begin{equation}\label{eq:criterion}
\sum_{x=1}^n \det \frac{\ad^*_r(A_{X_x|x})}{\langle 0|A_{X_x|x}|0\rangle}\geq n-1 \text{ for each } X_x\subset \mathcal{A}_x.
\end{equation}
\end{lemma} 
A proof of this result is given in Appendix \ref{Aps:jm-crit}.

Next we proceed to introduce the relevant measurements: we focus on the case of Alice attempting to steer Bob using rotated quadratures $Q_\theta=(e^{i\theta}a^\dagger + e^{-i\theta}a)/\sqrt 2$. They act in the infinite-dimensional Hilbert space, with spectral projections (PVM) 
$Q_{q|\theta} = e^{i\theta a^\dagger a}|q\rangle \langle q|e^{-i\theta a^\dagger a}$. As our state lives in ${\rm span}\, \{|0\rangle,|N\rangle\}$, only the 2x2-matrix 
$(\tilde Q_{q|\theta})_{nm} = \langle n| Q_{q|\theta}|m\rangle = e^{i\theta(n-m)} \langle n|q\rangle\langle q|m\rangle$ with $n,m=0,N$ contributes. Explicitly, this matrix reads
\begin{equation}
\tilde Q_{q|\theta} = \begin{pmatrix} 1 &  e^{-iN\theta}h(q)\\ e^{iN\theta}h(q) & h(q)^2 \end{pmatrix}\frac{e^{-q^2}}{\sqrt\pi}, \quad q\in \mathbb R,
\end{equation}
 where $h(x) := \frac{H_N(x)}{\sqrt{2^N N!}}$ with $H_N(x)$ a Hermite polynomial. Note that indeed $\int_{\mathbb R} \tilde Q_{q|\theta} dq = \id$ and $\tilde Q_{q|\theta}\geq 0$, so this is a valid qubit POVM with continuous outcomes.
 
We assume that Alice only has one pair, i.e.\ an assemblage $\{\tilde Q_{q|0}, \tilde Q_{q|\theta}\}_q$ for fixed $\theta$, and this pair is incompatible (despite the truncation), if $\theta\neq 0,\pi$. Indeed, since these are rank-1 POVMs, they can only be compatible if $\tilde Q_{q|\theta} = \int D(q|q')\tilde Q_{q|0}$ for some classical postprocessing $D(q|q')$ \cite{JP} implying $e^{i\theta}\in \mathbb R$, a contradiction. Hence the pure NOON-state (no damping, $r=1$) is steerable with these measurements.

The next step is then to compute the steering-equivalent (SE) observables by applying the channel $\ad_r$. With $0<r<1$, the SE observables become 
 \begin{equation}
\sch^*_r(\tilde Q_{q|\theta}) = \begin{pmatrix} 
1 & r e^{-iN\theta}h(q) \\
r e^{iN\theta}h(q) & r^2 h(q)^2+1-r^2 \end{pmatrix}
\frac{e^{-q^2}}{\sqrt\pi}.
\end{equation}
The determinant of this kernel matrix is $(1-r^2)\frac{e^{-2q^2}}{\pi}$ so that the joint measurability criterion of 
Lemma \ref{l:jmad} reduces to $r^2\leq 1/2$. From this we conclude that $${r_c \geq 1/\sqrt 2},$$ corresponding to ${\eta_c \geq 2/3}$ (independently of $\theta$ and $N$). The value $\eta_c \approx 2/3$ has previously been obtained numerically \cite{KoSk15} for $N=1$; up to our knowledge, ours is the first fully analytical proof of a lower bound on $\eta_c$.

We also remark that when $\eta\leq 2/3$, Eq.~(\ref{JMLHS}) used in the proof of Lemma \ref{l:jmad} gives an explicit joint observable and hence a local hidden state model preventing steering of $\rho_\eta$ by the two quadrature measurements.

Independently of Ref.~\cite{KoSk15}, we show that our method can provide also upper bounds on $r_c$ for $N=1$. We do this by binarising the POVMs, and recalling that incompatibility of binarisations is sufficient for that of the original POVMs, as coarse-graining is an instance of post-processing. Choosing the split at $q=0$ (i.e.\ Alice only records if $q>0$ or not) gives the POVM with elements $\frac 12(\id \pm {\bf n}\cdot \sigma)$ where ${\bf n}=(2r\sqrt 2 /\pi)(\cos\theta, \sin\theta,0)$. Using an exact criterion \cite{Busch} we conclude that the binarisations are incompatible for $r^2 \geq \pi (1-\sin\theta)/(2 \cos^2\theta)$. Notice that the bound depends on $\theta$; with $\theta = \pi/2$ (orthogonal quadratures) we get $r_c \leq \sqrt\pi /2$, or $\eta_c\leq 2\pi/(4+\pi)\approx 0.88$.

Since the split at $q=0$ is the most incompatible binarisation of quadratures \cite{HeKiRe15}, finer coarse-grainings are needed to get better bounds. By dividing the real line in $N_{\rm int}=2,4,6,8,10,12,14$ parts, we obtain bounds via SDP methods, cf.\ Fig.~\ref{fig:noisebounds} for pairs with varying $\theta$, and Tab.~\ref{tab:mult_int} for larger values of $N_{\rm int}$ with $\theta=\pi/2$. In particular, for $N_{\rm int}=20$ and $\theta=\pi/2$, we obtain the value $\eta_c \leq 0.671$, which is rather close to the lower bound $\eta_c\geq 2/3$.

\begin{figure}[t]
\begin{center}
\includegraphics[scale=0.7]{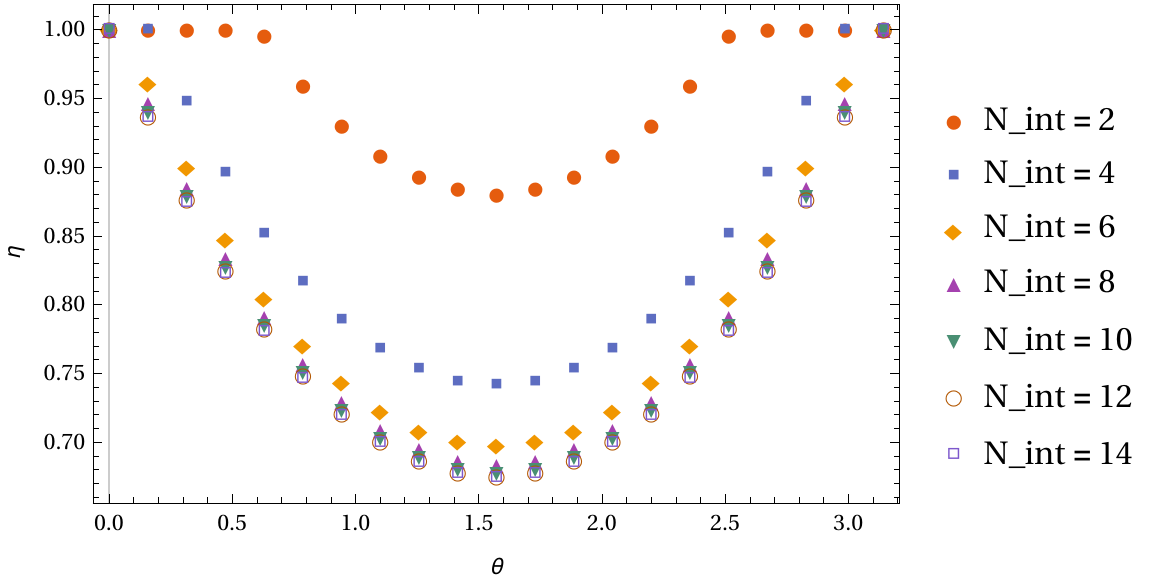}
\end{center}
\caption{\label{fig:noisebounds} (Color online) Critical noise bound for steering through the 1001-state by a coarse-grained pair of quadrature measurements, as a function of the separation angle $\theta$, with coarse-grainings of different number of intervals $N_{\rm int}$. The case $N_{\rm int}=2$ can be reproduced analytically. Below $\eta = 2/3$ the setting is unsteerable by joint measurability criterion \eqref{eq:criterion}.}
\end{figure}

\begin{table}[t!]
\begin{tabular}{|c | c |  c |  c |  c |  c |  c |  c |  c |  c |} \hline
$N_{\rm int}$  & 4 & 6 & 8 & 10 & 12 & 14 & 16 & 18 & 20\\ \hline
$\eta$ &   0.742 & 0.698 & 0.684 & 0.678 & 0.675 & 0.674 & 0.673 & 0.672 & 0.671 \\ \hline
   
\end{tabular}
\caption{Minimal $\eta$ such that the obtained $N_{\rm int}$-valued observables become incompatible.} \label{tab:mult_int}
\end{table}

We obtained these numerical results by implementing the SDP of the Incompatibility Robustness (IR) (see Eq.~\eqref{eq:IR}), searching for the values of $\eta_c$ for which ${\rm IR}>0$. We used the coarse-graining where $\mathbb R$ is divided into the intervals $(-\infty,-c]$, $[-c,-c +c/N_{\rm int}]$, $\ldots$, $[-c/N_{\rm int},0]$, $\ldots$, $[0,c/N_{\rm int}]$, $\ldots$, $[c,\infty)$, where $c\approx1.4$. The corresponding qubit observables were obtained by integrating over the intervals $I_k$, i.e., $Q_{I_k|\theta} = \int_{I_k} \tilde{Q}_{q|\theta} dq$; such integrals can be explicitly written in terms of error functions. 

One can try the same approach in the different subspaces with a higher number of photons. For instance, we investigated the case of $0$ or $6$ photons, which turned out to be more sensitive to noise, e.g., for the case of $N_{\rm int}=16$ one can reach $\eta_{\rm min}=0.89$. If one further increases the number of intervals, the computation becomes too slow and practically impossible.

\subsection{A dynamical example with non-Markovian noise}

We now illustrate how the above steering problem for the NOON state arises from a different context, and how our techniques provide a solution in that case as well.

Consider a setup where physical noise arises on Alice's side due to coupling to a zero temperature heat bath. Starting from the 1001-state, the photon dissipates into the bath on Alice's side via a channel $\mathcal E_t$ given by the amplitude damping master equation \cite{BP02}
$d\mathcal E_t(\rho_0)/dt=\gamma(t)\left[\sigma_{-}\mathcal E_t(\rho_0)\sigma_{+}-\frac{1}{2}\{\sigma_{+}\sigma_{-},\mathcal E_t(\rho_0)\}\right]$
where $\sigma_{+}=|1\rangle\langle 0|$, $\sigma_{-}=|0\rangle\langle 1|$, and $\gamma(t)=-2{\rm Re}\frac{d}{dt}\log G(t)$ with $G(t)$ depending on the bath spectral density. The state at time $t$ is $\rho_t = (\mathcal E_t\otimes {\rm Id})(|1001\rangle\langle 1001|)$ so by \eqref{eq:st_ch}, its channel $\sch=\sch_t$ equals $\mathcal E_t$ up to a unitary. Using the form of $\mathcal E_t$ \cite{AdHe15} we find $\sch_t^* (A) = U^*\ad_{r(t)}^*(A)U$, as in  Eq.~\eqref{eq:AD}, where now $r(t)=|G(t)|$, and $U$ is an irrelevant unitary. Interestingly, in this scenario our state-channel duality connects the steerability problem with the non-Markovian properties of the bath (cf., e.g., \cite{Li11}), previously associated with temporal correlations \cite{ChLa16} and decoherence of incompatibility \cite{AdHe15}.

The result of the preceding subsection can now be directly applied to characterise steering in the heat bath scenario: for any time $t$, the state $\rho_t$ is steerable by $\{Q_{q|0},Q_{q|\pi/2}\}$ iff $r(t)\geq r_c$. For the typical Lorentzian spectral density, $r(t)= e^{-\lambda t/2} |\cosh(w \lambda t /2)+\sinh(w \lambda t /2)/w|$ where $\lambda $ is the linewidth, and $w=\sqrt{1-2u/\lambda}$ with $u$ the coupling strength \cite{AdHe15}. We can then evaluate $r(t)\geq r_c$ with the numerical value $r_c \approx 1/\sqrt{2}$, to get the region of points $(u,t)$ where the state is steerable; cf.\ Fig. \ref{fig:region} and its caption. 

\begin{figure}
\begin{center}
\includegraphics[scale=0.22]{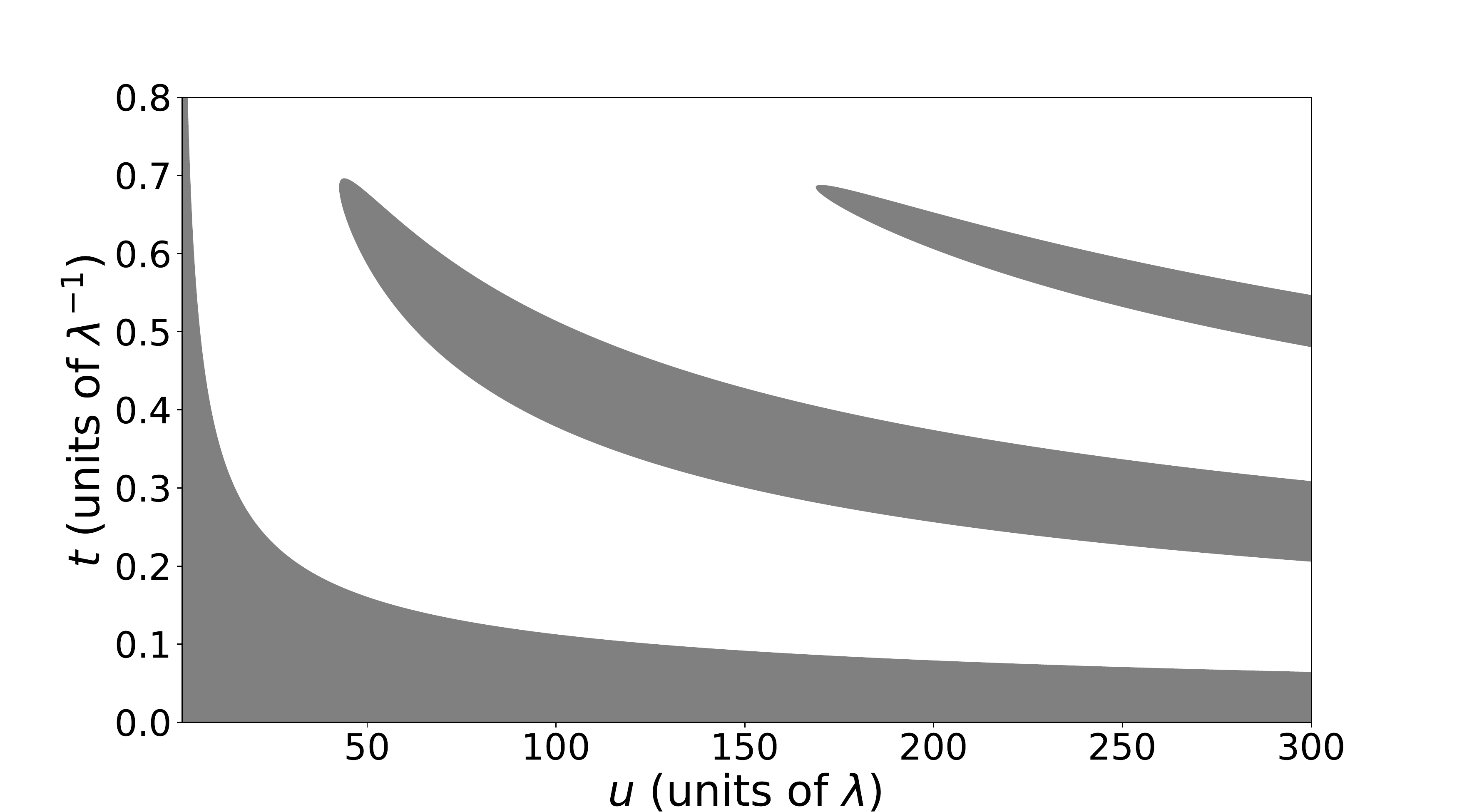}
\end{center}
\caption{\label{fig:region}  (Color online) Steerability region for the dynamical setting (shaded area). The parameter $u$ is the coupling strength (in units of the spectral linewidth $\lambda$ of the bath) and $t$ is time (in units of $\lambda^{-1}$). The two revival regions reflect the non-Markovian character of the evolution in the strong coupling regime, which allows steerability to re-emerge at later times.}
\end{figure}

\section{Gaussian case}\label{sec:gauss}

In this section we establish the correspondence between steering of Gaussian states and incompatibility of Gaussian measurements, via a Gaussian version of our general state-channel duality. In order to do this we first need to establish the required formalism and introduce the notation.

Starting with the basics, an optical system with $N$ modes is a continuous variable (CV) quantum system (see e.g.\ \cite{SiMu94}) with the infinite-dimensional Hilbert space $\mathcal H^{\otimes N}  = \otimes_{j=1}^N L^2(\mathbb R) \simeq L^2(\mathbb R^N)$. The associated phase space is $\mathbb R^{2N}$, with canonical coordinates ${\bf x}=(q_1,p_1,\ldots, q_N,p_N)^T$ in a fixed symplectic basis. The corresponding standard quadrature operators are denoted by $Q_j$ and $P_j$; they satisfy $[Q_i,P_j]=i\delta_{ij}\id$, $[Q_i,Q_j]=[P_i,P_j]=0$  and we set ${\bf R} = (Q_1,P_1,\ldots, Q_N,P_N)^T$, so that $[R_i,R_j]=i\Omega_{ij}\id$ with
${\bf \Omega} = \oplus_{j=1}^N \left( \begin{smallmatrix} 0 & 1 \\ -1 & 0\end{smallmatrix}\right)$. We further denote $$Q_{\bf x}={\bf x}^T{\bf R};$$
these operators are called (generalised) \emph{quadratures}. For a pair of quadratures $(Q_{\bf x},Q_{\bf y})$ the commutator is given by $[Q_{\bf x},P_{\bf y}]=i{\bf x}^T{\bf \Omega}{\bf y}\id$, and any pair for which ${\bf x}^T{\bf \Omega}{\bf y}=1$ is called \emph{canonical}.

The \emph{Weyl operators} $W({\bf x}) = e^{-i Q_{\bf x}}$ satisfy the canonical commutation relation (CCR)
\begin{equation}\label{eqn:commutation}
W({\bf x})W({\bf y}) = e^{-i{\bf x}^T {\bf \Omega y}} W({\bf y})W({\bf x}) \, ,
\end{equation}
and we define displacement operators $D_{{\bf c}}:=W({\bf \Omega}^T{\bf c})$ so that $D_{\bf c}^*W({\bf x})D_{\bf c}=e^{-i{\bf c}^T{\bf x}}W({\bf x})$. A matrix ${\bf S}$ is \emph{symplectic} if ${\bf S}^T{\bf \Omega}{\bf S}={\bf \Omega}$; then by Stone-von Neumann theorem there is a unitary $U_{\bf S}$ with $U_{\bf S}^*W({\bf x})U_{\bf S}=W({\bf Sx})$.

\subsection{Gaussian states, measurements, channels, and postprocessings}

In the following, we first review the characteristic function formalism for Gaussian quantum objects \cite{SiMu94,HoWe01}; see also \cite{KiSc13,gibc}. We then use this to prove the Gaussian version of the state-channel correspondence, after which we proceed to establish the connection between steering and incompatibility.

The characteristic function formalism treats Gaussian state, channels, measurements, and postprocessings in the same footing, is a transparent quantum analogue of the corresponding classical objects by way of a rigorous correspondence theory \cite{We84}, does not require the use of ancillas, circumvents the technical problem of the POVM elements not always being proper operators (see the discussion above), and is especially convenient to use with concatenation, making explicit the idea that a Gaussian channel applied to a Gaussian state (Schrodinger picture) or measurement (Heisenberg picture) produces a new Gaussian state and measurement, respectively. We note that this approach differs from the alternative (equivalent) one introduced by Giedke and Cirac \cite{Giedke02}, on which Wiseman {\it et al.} based their derivation of the Gaussian steering criterion \cite{WiJoDo07}.

A state on a CV system is \emph{Gaussian} if its characteristic function $\hat{\rho}({\bf x}):={\rm tr}[\rho W({\bf x})]$ is a Gaussian function:
\begin{equation}\label{covdef}
\hat{\rho}({\bf x}) = e^{-\frac{1}{4}{\bf x}^T {\bf V}_\rho {\bf x} - i{\bf r}^T {\bf x}}  
\end{equation}
where ${\bf V}_\rho$ is the \emph{covariance matrix} (CM) $[V_\rho]_{ij} = {\rm tr}[\rho\{ R_i - r_i,R_j-r_j \}]$ with displacement vector $r_j={\rm tr}[\rho R_j]$. The CM satisfies the \emph{uncertainty relation}
\begin{equation}\label{UR}
{\bf V}_\rho + i{\bf \Omega }\geq 0.
\end{equation}
Crucially, \emph{every} real and symmetric matrix ${\bf V}$ satisfying \eqref{UR} is a CM of some Gaussian state $\rho$.


A measurement (POVM) $M_{\bf a}$ with outcomes ${\bf a}\in \mathbb R^{d}$, is \emph{Gaussian} if its outcome distribution for any Gaussian state is a Gaussian (i.e.\ Normal) distribution.
This is the case when the operator-valued characteristic function $\hat M({\bf p}):=\int e^{i{\bf p}^T{\bf a}}M_{\bf a}\,d{\bf a}$ is of the form
\begin{equation}\label{eqn:gaussian_obs}
\hat M({\bf p})= W({\bf Kp})e^{-\frac 14 {\bf p}^T{\bf Lp}- i{\bf m}^T{\bf p}},
\end{equation}
where ${\bf K}$ is an $N\times d$-matrix and ${\bf L}$ is an $d\times d$-matrix satisfying the positivity condition
\begin{equation}\label{possu}
{\bf C}_{\bf K,L} :={\bf L} -i {\bf K}^T{\bf \Omega}{\bf K}\geq 0,
\end{equation}
and ${\bf m}$ is a displacement vector. Importantly, \emph{every} triple $({\bf K},{\bf L}, {\bf m})$ satisfying \eqref{possu} defines a Gaussian measurement.

In the case $d=1$ we have ${\bf K}={\bf x}$, a column vector, while ${\bf L}=2\xi^2$ and ${\bf m}=m$ are just numbers. Since shifts in outcomes are irrelevant for steering, we consider $m=0$ so that the corresponding POVM $M_{a|{\bf x},\xi}$ has characteristic function
$$\hat M_{a|{\bf x},\xi}(p)=e^{-ipQ_{\bf x}} e^{-\frac 12 p^2\xi^2}.$$ With $\xi^2=0$ we simply obtain the PVM with characteristic function $\hat M(p)=e^{-ipQ_{\bf x}}$, that is, the unitary group generated by the quadrature operator $Q_{\bf x}$. Consistently with the notation in previous section, we use $Q_{a|{\bf x}}$ to denote the corresponding PVM elements. Hence, Gaussian PVMs with $d=1$ are just quadrature measurements. In general, the product form of the characteristic function implies that $M_{a|{\bf x},\xi}$ has the convolution form
\cite{KiSc13}:
$$M_a=M_{a|{\bf x},\xi} := \frac{1}{\xi\sqrt{2\pi}}\int e^{-\frac 12 (a-a')^2/\xi^2} Q_{a'|{\bf x}}da'.$$ 
Hence, any Gaussian POVM $M_a$ with $a\in\mathbb R$ is, up to a shift, a ``noisy" quadrature. Interestingly, noise exceeding the uncertainty limit renders quadratures jointly measurable:
\begin{lemma}\label{l:gc}
The noisy versions $M_{{\bf x},\xi}$ and $M_{{\bf y},\xi'}$ of two quadratures $Q_{\bf x}$, $Q_{\bf y}$ are jointly measurable if and only if $$\xi\xi'\geq \|[Q_{\bf x},P_{\bf y}]\|/2,$$
in which case they have a Gaussian joint measurement.
\end{lemma}
This result generalises a known joint measurability criterion for position and momentum \cite{CaHeTo05,BuHeLa07,We04}; see Appendix \ref{Aps:jm-q} for a proof. The crucial point here is the existence of joint Gaussian measurement, which follows from the nontrivial averaging argument of \cite{We04}.


A quantum channel between two CV systems with respective degrees of freedom $N$ and $N'$ is \emph{Gaussian}, if it maps Gaussian states into Gaussian states. In the Heisenberg picture, this entails 
\begin{equation}\label{eqn:gaussian_channel}
\Lambda^*(W({\bf x})) = W({\bf Mx})e^{-\frac{1}{4}{\bf x}^T {\bf Nx} - i{\bf c}^T{\bf x}}
\end{equation}
where ${\bf M}$ is a real $2N\times 2N'$-matrix, and ${\bf N}$ is a real $2N'\times 2N'$-matrix. Due to complete positivity, they satisfy 
\begin{equation}\label{eqn:cp_channel}
{\bf C}_{\bf M,N}+ i {\bf \Omega}\geq {\bf 0},
\end{equation}
where (interestingly) ${\bf C}_{\bf M,N}$ is as in \eqref{possu}. Again, \emph{every} triple $({\bf M}, {\bf N}, {\bf c})$ with \eqref{eqn:cp_channel} defines a Gaussian channel via \eqref{eqn:gaussian_channel}. Unitary channels $B\mapsto U^*BU$ have ${\bf N}={\bf 0}$ and ${\bf M}={\bf S}$ symplectic, i.e. $U=D_{\bf c}U_{\bf S}$. Using \eqref{covdef} and \eqref{eqn:gaussian_channel} we get the general transformation rule for states in terms of CMs and displacement vectors:
\begin{equation}\label{channelcov}
{\bf V}\mapsto {\bf M}^T {\bf V}{\bf M}+{\bf N},\quad {\bf r}\mapsto {\bf M}^T{\bf r}+ {\bf c}.
\end{equation}

Similarly, a Gaussian channel with matrices $({\bf M},{\bf N},{\bf c})$, followed by a Gaussian measurement with matrices $({\bf K}, {\bf L}, {\bf m})$ is clearly a Gaussian measurement as well, and we can easily derive the associated matrices by combining \eqref{eqn:gaussian_obs} and \eqref{eqn:gaussian_channel}; there the result is
\begin{equation}\label{eqn:ch_meas}
({\bf K},{\bf L}, {\bf m})\mapsto ({\bf MK}, {\bf L}+{\bf K}^T{\bf NK}, {\bf m}+{\bf K}^T{\bf c}).
\end{equation}
Using \eqref{eqn:ch_meas} we observe that (for ${\bf c}=0$) the channel transforms a quadrature PVM $Q_{{\bf x}}$ into the noisy POVM $M_{{\bf Mx},\xi}$ where now $\xi^2= {\bf x}^T{\bf Nx}/2$.

Finally, a Gaussian post-processing (\emph{classical} channel) is one which transforms every Gaussian probability distribution into another one. These are determined by triples $({\bf M},{\bf N},{\bf c})$ as in the above quantum case, \emph{except that only $\bf N\geq 0$ is required} as complete positivity does not appear in the classical case. One can show that the matrices are associated with linear coordinate transformations, convolutions, and translations, respectively \cite{gibc}. Note that linear transformations include the deterministic post-processings which simply project on a lower-dimensional subspace. A Gaussian measurement $({\bf K}, {\bf L}, {\bf m})$, followed by a Gaussian postprocessing $({\bf M},{\bf N},{\bf c})$, is again a Gaussian measurement, with parameters obtained by the transformation rule
\begin{equation}\label{eqn:meas_post}
({\bf K},{\bf L}, {\bf m})\mapsto ({\bf KM}, {\bf N}+{\bf M}^T{\bf LM}, {\bf c}+{\bf M}^T{\bf m}).
\end{equation}

\subsection{State-channel correspondence and Gaussian steering}
We are now ready to prove our main results on Gaussian steering. We start with the Gaussian version of the state-channel duality:

\begin{lemma}\label{l:gau_cj} There is a 1-to-1 correspondence between bipartite Gaussian states $\rho$ sharing a marginal $\sigma = {\rm tr}_A[\rho]$ with CM ${\bf V}_\sigma$ of full symplectic rank and displacement ${\bf r}_\sigma$, and Gaussian channels $\sch$ from Bob to Alice, such that \eqref{eq:st_ch} holds with $|\Omega\rangle$ having CM and displacement
$$
{{\bf V}_{\Omega} = \begin{pmatrix}
{\bf V}_\sigma & {\bf S}^T{\bf Z}{\bf S} \\
{\bf S}^T{\bf Z}{\bf S} & {\bf V}_\sigma
\end{pmatrix}}, \quad {\bf r}_{\Omega} = {\bf r}_\sigma\oplus {\bf r}_\sigma.
$$ Here ${\bf S}$ is a symplectic matrix diagonalising ${\bf V}_\sigma$, and ${\bf Z} = \oplus_{i=1}^N\sqrt{\nu_i^2-1} \,\sigma_z$, with $\nu_i$ the symplectic eigenvalues of ${\bf V}_\sigma$. The correspondence between the parameters $({\bf V},{\bf r})$ and $(\bf M,\bf N,c)$ of $\rho$ and $\sch$, respectively, is explicitly given by
\begin{align*}
\left\{\begin{array}{ll}
{\bf V} &= \left( \begin{array}{cc}
{\bf V}_A & {\bf \Gamma}^T \\
{\bf \Gamma} & {\bf V}_\sigma
\end{array}\right),\\
{\bf r} &= {\bf r}_A\oplus {\bf r}_B\end{array}\right.
 &\leftrightarrow \left\{\begin{array}{ll} {\bf M}&= ({\bf S}^T{\bf Z}{\bf S})^{-1}{\bf \Gamma}\\
{\bf N}&= {\bf V}_A-{\bf M}^T{\bf V}_\sigma {\bf M}\\
{\bf c} &= {\bf r}_A-{\bf M}^T{\bf r}_\sigma \end{array}\right.,
\end{align*}
where the positivity conditions are equivalent: ${\bf V} +i{\bf \Omega}\geq 0$ iff ${\bf C}_{\bf M,N}+ i{\bf \Omega}\geq 0$.
\end{lemma}
The proof of this Lemma is given in Appendix \ref{Aps:gauss-dual}. Interestingly, the equivalence of the inequalities is obtained via Schur complements, which have recently found applications in the investigation of quantum correlations \cite{lami}. Using the Lemmas \ref{l:gc} and \ref{l:gau_cj} we finally prove
\begin{theorem}\label{th:gau_st}
Let $\rho$ be a bipartite Gaussian state with CM ${\bf V}$ and displacement ${\bf r}$, and $({\bf M, N,c})$ the matrices of the channel $\sch$ given by Lemma \ref{l:gau_cj}. The following are equivalent:\\
\noindent{}(i) $\rho$ is steerable by the set of Gaussian measurements.\\
\noindent{}(ii) $\rho$ is steerable by some canonical pair of quadratures.\\
\noindent{}(iii) ${\bf V} + i({\bf 0}\oplus {\bf \Omega})$ is not positive semidefinite.\\
\noindent{}(iv) $({\bf M,N,c})$ do not define a valid Gaussian observable.\\
\end{theorem}
\begin{proof}
We first note that (ii) trivially implies (i). Next, we repeat the calculation \eqref{schur} in the proof of Lemma \ref{l:gau_cj} (see Appendix \ref{Aps:gauss-dual}) without ${\bf \Omega}_A$, which establishes that ${\bf C}_{\bf M,\bf N}$ is the Schur complement of ${\bf V}_{\sigma}+i{\bf \Omega}_B$ in ${\bf V}_\rho+i({\bf 0}\oplus {\bf \Omega}_B)$. This shows that (iii) and (iv) are equivalent. Furthermore, using \cite[Prop.\ 2]{gibc} we conclude that $\sch$ maps the set of all Gaussian measurements into a set having a joint (Gaussian) measurement, if ${\bf C}_{\bf M,N}\geq 0$. Hence (i) implies (iv).

We are left with the proof of the main result, stating that (iv) implies (ii). Assuming (iv) let ${\bf x},{\bf y}$ be vectors such that
$({\bf y}^T-i{\bf x}^T){\bf C}_{\bf M,N}({\bf y}+i{\bf x})<0$. Then by complete positivity $({\bf y}^T-i{\bf x}^T)({\bf C}_{\bf M,N}+i{\bf \Omega})({\bf y}+i{\bf x})\geq 0$, which implies $r:={\bf x}^T{\bf \Omega}{\bf y}>0$ and
\begin{equation}\label{urc}
({\bf Mx})^T{\bf \Omega}{\bf My}>\frac 12({\bf x}^T{\bf N}{\bf x} + {\bf y}^T{\bf N}{\bf y}).
\end{equation}
Clearly, we may replace ${\bf x}$ and ${\bf y}$ with $r^{-\frac 12} {\bf x}$ and $r^{-\frac 12}{\bf y}$ and \eqref{urc} still holds. Then the pair $Q_{\bf x}={\bf x}^T{\bf R}$ and $P_{\bf y}={\bf y}^T{\bf R}$ of quadratures is canonical since ${\bf x}^T{\bf \Omega}{\bf y}=1$. It is easy to check using the transformation rule \eqref{eqn:ch_meas} that the channel $\sch$, having parameters $({\bf M,N,c})$, transforms the associated PVMs into the POVMs $M_{{\bf Mx},\xi}$ and $M_{{\bf My},\xi'}$ (up to irrelevant shifts in outcomes) where $\xi^2= {\bf x}^T{\bf Nx}/2$ and $\xi'^2={\bf y}^T{\bf Ny}/2$. By \eqref{urc} we have $2\xi\xi' \leq \xi^2+\xi'^2< ({\bf Mx})^T{\bf \Omega}{\bf My}$ so from Lemma \ref{l:gc} we conclude that the POVMs are not jointly measurable. This means we have found a canonical pair $(Q_{\bf x}, Q_{\bf y})$ of quadratures such that $\sch(Q_{\bf x})$ and $\sch(Q_{\bf y})$ are not jointly measurable, so according to Th.~\ref{th:asmb}, the state $\rho$ is steerable by this pair. Hence, (ii) holds. The proof is complete.
\end{proof}

We remark that the equivalence between (i) and (iii) was originally proven in \cite{WiJoDo07}. Here we use Lemma \ref{l:gc} to show that quadratures are enough [(ii)]; this comes closest to the original notion of steering of an EPR-state via position and momentum as discussed by Schr\"odinger \cite{Schr35}.  Note that the above proof shows explicitly how one can construct quadrature pairs for which steering is possible when the conditions of the theorem hold.

Furthermore, a new interpretation emerges from (iv): the Gaussian POVM determined by the \emph{channel parameters} $({\bf M,N,c})$ is exactly the joint observable for the assemblage $\{\mathsf T^*(A_{\bf a}): A_{\bf a} \text{ Gaussian}\}$ that rules out steering in (i) by Th.~\ref{th:asmb}. In order to explain this in detail, we follow the argument in \cite[Prop. 2]{gibc} mentioned in the above proof: we first note that an arbitrary Gaussian measurement $({\bf K,L,m})$ on Alice's side is transformed by the channel $({\bf M,N,c})$ into one with parameters $({\bf K}',{\bf L}', {\bf m}')=({\bf MK}, {\bf L}+{\bf K}^T{\bf NK}, {\bf m}+{\bf K}^T{\bf c})$ by \eqref{eqn:ch_meas}. In order to show that such POVMs are all jointly measurable, we only need to reinterpret the channel parameters $({\bf M,N,c})$ as the joint measurement $G_{\boldsymbol \lambda}$. Indeed, with $({\bf K,L,m})$ taken as postprocessing parameters, \eqref{eqn:meas_post} becomes identical to \eqref{eqn:ch_meas}, showing how $({\bf K}',{\bf L}', {\bf m}')$ is postprocessed from $G_{\boldsymbol \lambda}$. We stress that the nontrivial part is in the positivity requirements, which are \emph{not} in general identical. Indeed, the reinterpretation is possible only when (iv) does not hold, i.e. ${\bf C}_{\bf M,N}\geq {\bf 0}$, which is not true for general channels.

For the sake of completeness, and in order to further demonstrate that the existing formulation of Gaussian steering \cite{WiJoDo07} follows from our theory, we also show how one can easily derive the Gaussian LHS model given in \cite{WiJoDo07} from the above results. Since this is not essential for understanding our main results, the derivation is given in Appendix \ref{App:LHS}.

Finally, in addition to its impact on Gaussian steering, Thm.~\ref{th:gau_st} yields
\begin{corollary}
A Gaussian channel which maps each canonical quadrature pair into a jointly measurable pair, is \emph{Gaussian incompatibility breaking} in the sense of \cite{gibc}.
\end{corollary}
This considerably strengthens the theory in \cite{gibc}, by showing that \emph{canonical} pairs are sufficient, and that a \emph{Gaussian} joint observable always exists for the jointly measurable Gaussian POVMs. The latter is a nontrivial and a fairly fundamental result which requires Lemma \ref{l:gc}.

\section{Conclusions} Steering is a genuine quantum phenomenon, with important applications both in quantum information processing and foundations of quantum mechanics. Notwithstanding the growing interest in it in the past few years \cite{CaSk16b}, limited results and tools are available in the continuous variable case. We introduced a state-channel correspondence that allows us to discuss the steering problem in a completely general context. In particular, we extend many of the results previously known only in the finite-dimensional case, such as the mathematical equivalence of steering and joint-measurability problems \cite{UoBu15} and the equivalence of steering and joint-measurability for the case of full Schmidt rank states \cite{UoMoGu14,QuVeBr14}.
Moreover, via state-channel duality we are able to connect steerability properties of noisy NOON states with Markovianity properties of the corresponding channel, and to provide an analytical lower bound to the steerability noise threshold for any $N$. Finally, we apply our methods to the Gaussian setting, introducing a new channel characterization of steerability and proving that canonical quadratures are enough for steering. An interesting future direction would be to extensively investigate the capability of the state-channel duality to provide steering, joint measurability, and incompatibility breaking criteria in the continuous variable case for states, observables, and channels, respectively.

\section*{Acknowledgements} The authors thank O. G\"uhne and T. Heinosaari for fruitful discussions. This work has been supported by FWF (Project: M 2107 Meitner-Programm), EPSRC (project EP/M01634X/1), and the Finnish Cultural Foundation.

\appendix

\section{Proof of the general state-channel duality (Lemma \ref{l:channel})}\label{Aps:sc-dual}
Let $|\Omega_\sigma\rangle$ and $\sigma$ be as in Lemma \ref{l:channel}. Given any quantum channel $\sch$, the state
\begin{equation}\label{ch1}
\rho = (\sch \otimes {\rm Id})(|\Omega_\sigma\rangle\langle \Omega_\sigma|)
\end{equation}
clearly has the property ${\rm tr}_A[\rho]=\sigma$, so we have managed to produce more general states than ones obtained by the Choi-Jamio\l{}kowski correspondence. We now need to prove that the new correspondence is one-to-one \emph{onto the set of states with ${\rm tr}_A[\rho]=\sigma$}. Given such as state, we first compute
\begin{align}
\nonumber{\rm tr}[\rho (A\otimes B)] &=\langle \Omega_\sigma |\sch^*(A)\otimes B|\Omega_\sigma\rangle\\
\nonumber &= \sum_{nm} \sqrt{s_ns_m} \langle nn|\sch^*(A)\otimes B|mm\rangle\\
\nonumber &= \sum_{nm}\sqrt{s_ns_m} \langle n|\sch^*(A)|m\rangle\langle n|B|m\rangle\\
\nonumber &= \sum_{nm}\langle n|\sqrt{\sigma}\sch^*(A)\sqrt{\sigma}|m\rangle\langle n|B|m\rangle\\
 & = {\rm tr}[\sqrt{\sigma}\sch^*(A)\sqrt{\sigma}B^\intercal],
\end{align}
where $B^\intercal$ is the transpose of $B$ in the fixed basis. Hence,
\begin{equation}\label{ch}
\sigma^{\frac 12} \sch^*(A)\sigma^{\frac 12}={\rm tr}_A[\rho(A\otimes \id)]^\intercal.
\end{equation}
From this we see immediately that distinct channels correspond to distinct states, since the matrix elements of the state are clearly uniquely determined by those of the channel: $\langle nm|\rho|n'm'\rangle = {\rm tr}[\sqrt{\sigma}\sch^*(|n'\rangle \langle n|)\sqrt{\sigma}(|m'\rangle\langle m|)^\intercal] = \sqrt{s_m}\sqrt{s_{m'}} \langle m'| \sch^*(|n'\rangle \langle n|)|m\rangle$,  where we have now also fixed a basis $\{|n\rangle\}$ on Alice's side.

What remains to be shown is that for \emph{any} state $\rho$ with ${\rm tr}_A[\rho]=\sigma$ there exists a channel $\sch$ such that \eqref{ch1} (or, equivalently, \eqref{ch}) holds. If $d<\infty$ we can invert $\sigma^{-\frac 12}$ in \eqref{ch} to solve for $\sch^*(A)$; however, we still need to show that this defines a channel, i.e.\ a CPTP map. We therefore proceed by writing the state $\rho$ as
\begin{align}
\rho & = \sum_k |\psi_k\rangle\langle \psi_k|  = \sum_{\substack{k,n,m\\ n',m'}}\langle nm|\psi_k\rangle\langle \psi_k| n'm'\rangle\, |nm\rangle\langle n'm'|
\end{align}
so that, for all bounded operators $A$ and $B$ (for Alice and Bob, respectively), we get
\begin{align}
{\rm tr}[\rho (A\otimes B)] &=\sum_k {\rm tr}[R_k^*AR_kB^\intercal],
\end{align}
where $R:\mathcal H_B\to \mathcal H_A$ is the Hilbert-Schmidt operator defined by $\langle n|R_k m\rangle =\langle nm|\psi_k\rangle$. Hence, ${\rm tr}_A[\rho (A\otimes \id)]^\intercal=\sum_k R_k^*AR_k$. In particular, $\sigma=\sigma^\intercal=\sum_k R_k^*R_k$.

Next, we need a little of functional analysis, so as to allow the proof to go through also for $d=\infty$, in which case the inverse of any full rank state is unbounded and requires some care. Let $\mathcal R$ be the dense range of $\sigma$, containing all the basis vectors. Then $\mathcal R={\rm ran}\, \sigma^{\frac 12}$, $\sigma^{\frac 12}$ is injective, and for any $|\psi\rangle\in \mathcal R$ we have $$\|R_k\sigma^{-\frac 12}\psi\|^2\leq \sum_k \langle \sigma^{-\frac 12}\psi|R_k^*R_k\sigma^{-\frac 12}\psi\rangle = \|\psi\|^2,$$
which implies that each $R_k\sigma^{-\frac 12}$ extends to a bounded operator $M_k:\mathcal H_B\to\mathcal H_A$, for which $M_k\sigma^{\frac 12}=R_k$.

Since $\sum_k M_k^*M_k = \id$, the operators $M_k$ set up a Kraus decomposition of a channel: we define
\begin{equation}
\sch(T) := \sum_k M_kTM^*_k
\end{equation}
for all (trace class) operators $T$. This is by construction completely positive, and it is trace-preserving since $\sum_k M_k^*M_k = \id$. In the infinite-dimensional case the series converges, e.g., in the weak topology. Plugging this channel in Eq.~\eqref{ch} immediately gives
\begin{align}
\nonumber
\sigma^{\frac 12} \sch^*(A)\sigma^{\frac 12} &= \sum_k (M_k\sigma^{\frac 12})^*AM_k\sigma^{\frac 12}= \sum_k R_k^*AR_k\\
& = {\rm tr}_A[\rho (A\otimes \id)]^\intercal,
\end{align}
so that \eqref{ch}, and hence also \eqref{ch1} holds, that is, the channel gives back the original state $\rho$. This proves that the correspondence is one-to-one, and completes the proof.

\section{A joint measurability criterion for qubit POVMs with arbitrary outcomes (Lemma \ref{l:jmad})}\label{Aps:jm-crit}

In contrast to most existing criteria, this one applies to qubit POVMs with continuous outcome sets. In the main text, it was shown to be useful for finding noise bounds for quadrature measurements restricted to two-dimensional photon number eigenspaces. 

More generally, we prove that an assemblage of $n$ qubit observables $\{B_{b|i} \}_{i=1}^n$ is jointly measurable if
\begin{equation}\label{eq:apcrit}
\Delta(b_1,\ldots,b_n) :=\sum_i r_i(b_i)-n+1\geq 0,
\end{equation}
where $r_i(b):=\det M_i(b)$, $M_i(b) := B_{b|i}/p_i(b)$, and $p_i(b) = \langle 0|B_{b|i} |0\rangle$. Indeed,
\begin{equation}
M_i(b) =\begin{pmatrix} 1  & \overline{f_i(b)}\\ f_i(b) & r_i(b) + |f_i(b)|^2\end{pmatrix}
\end{equation}
for some complex valued functions $f_i$. The normalisation forces $\int f_i(b) p_i(x)dx =0$ and $\int (|f_i(b)|^2 +r_i(b)) p_i(b)db =1$. We define $G_{b_1,\ldots,b_n}$ via
\begin{equation}\label{JMLHS}
\frac{G_{b_1,\ldots,b_n}}{\prod_{i=1}^n p_i(b_i)} = \begin{pmatrix} 1  & \sum_i\overline{f_i(b_i)}\\ \sum_i f_i(b_i) & |\sum_i f_i(b_i)|^2+\Delta(b_1,\ldots,b_n)\end{pmatrix}.
\end{equation}
Using the constraints we see that it is normalised, and that
\begin{equation}
B_{b|i}=M_i(b)p_i(b) = \int \delta_{b,b_i}G_{b_1,\ldots,b_n} db_1\cdots db_n.
\end{equation}
The critical constraint is $G_{b_1,\ldots,b_n}\geq 0$ now follows from
\begin{equation}
\det G_{b_1,\ldots,b_n}=\Delta(b_1,\ldots,b_n)\prod_{i=1}^n p_i(b_i)^2 \geq 0,
\end{equation}
which is ensured by the assumption. This means that the $B_i$ have a joint observable with deterministic response functions, so they are jointly measurable.

By taking $B_{a|i} = \ad_r^*(A_{a|i})$, where $\ad_r$ is the amplitude damping channel defined in the main text, the assumption Eq.~\eqref{eq:apcrit} becomes Eq.~\eqref{eq:criterion} of the main text, once we notice that $\langle 0|\ad_r^*(A_{a|i})|0\rangle = \langle 0|A_{a|i}|0\rangle$; see \eqref{eq:damp_matrix}. This completes the proof.

\section{Proof of the joint measurability criterion for convoluted quadratures (Lemma \ref{l:gc})}\label{Aps:jm-q}
This lemma was critical for the characterisation of Gaussian steering. In order to prove it we let 
$r={\bf x}^T{\bf \Omega}{\bf y}$, so that $[Q_{\bf x},Q_{\bf y}] = ir\id$. If $r=0$ then $Q_{\bf x}$ and $Q_{\bf y}$ commute and the claim is trivial 
since they stay jointly measurable after convolution. We suppose $r>0$, and look at the scaled quadrature $Q_{\bf y}/r= {\bf y}^T{\bf R}/r= Q_{{\bf y}/r}$. By using the connection $Q_{\bf y}= \int a\, Q_{a|{\bf y}}da$ between the operator $Q_{\bf y}$ and the corresponding PVM $Q_{a|{\bf y}}$, we see that $Q_{a|{\bf y}/r}= rQ_{ra|{\bf y}}$. A direct computation then shows that scaling of the noisy POVM gives $M_{a|{\bf y}/r, \xi'/r}=r M_{ra|{\bf y},\xi'}$. Since scaling is a postprocessing and hence does not affect joint measurability, the original pair $(M_{{\bf x},\xi}, M_{{\bf y}, \xi'})$ is jointly measurable if and only if $(M_{{\bf x},\xi}, M_{{\bf y}/r, \xi'/r})$ is. But the corresponding quadrature pair $(Q_{\bf x},Q_{{\bf y}/r})$ is canonical, as $$[Q_{\bf x}, Q_{{\bf y} /r}]= [{\bf x}^T{\bf R},{\bf y}^T{\bf R}/r]= i({\bf x}^T{\bf \Omega}{\bf y}/r)\id=i\id,$$ and hence unitarily equivalent to the pair $(Q_0,Q_{\pi/2})$ via a symplectic transformation, where $Q_{\theta}= (e^{i\theta}a^\dagger+e^{-i\theta}a)/\sqrt{2}$ are the rotated quadratures of a single-mode system. The same unitary then transforms the convoluted pair $(M_{{\bf x},\xi}, M_{{\bf y}/r, \xi'/r})$ into the pair $(M_{0,\xi},M_{\pi/2,\xi'/r})$ where
$$M_{a|\theta, \xi} := \frac{1}{\sqrt{2\pi}\xi}\int e^{-\frac 12 (a-a')^2/\xi^2} Q_{a'|\theta},$$ and hence it suffices to show that the joint measurability of $(M_{0,\xi},M_{\pi/2,\xi'/r})$ is equivalent to the inequality $\xi(\xi'/r)\geq 1/2$, and that the joint observable, when exists, can be chosen Gaussian.

In order to prove this, we use known results on joint measurability of ``unsharp'' position and momentum \cite{CaHeTo05,BuHeLa07}, which is exactly what our convoluted quadratures are. In particular, if $(M_{0,\xi},M_{\pi/2,\xi'/r})$ are jointly measurable, they must have a joint observable of the Weyl-covariant form $G_{a_1,a_2}= W(a_1,a_2)\rho_0W(a_1,a_2)^*/(2\pi)$, where $\rho_0$ is a state with
\begin{align}
{\rm tr}[\rho_0Q_{a_1|0}]&= \frac{e^{-\frac 12 a_1^2/\xi^2}}{\sqrt{2\pi}\xi}, & 
{\rm tr}[\rho_0 Q_{a_2|\pi/2}]&=\frac{e^{-\frac 12 a_2^2/(\xi'/r)^2}}{\sqrt{2\pi}(\xi'/r)}.
\end{align}
This implies that $\xi$ and $\xi'/r$ are the standard deviations of $Q_0$ and $Q_{\pi/2}$ in the state $\rho_0$, hence satisfying $\xi(\xi'/r)\geq 1/2$ by the Heisenberg uncertainty principle. Conversely, if the inequality holds, we can define $\rho_0=|\psi_0\rangle\langle \psi_0|$ in the coordinate representation as $\psi_0(a)=(2c/\pi)^{\frac 14} e^{-(c+iw)a^2}$ with $\xi^2=1/(4c)$ and $\xi'^2/r^2=(c^2+d^2)/d$; then a direct computation shows that the corresponding $G_{a_1,a_2}$ is a joint observable for $M_{0,\xi}$ and $M_{\pi/2,\xi'/r}$. This observable is Gaussian since $\rho_0$ is a Gaussian state \cite{KiSc13}.

Finally, since all the above unitary equivalences were done via symplectic transformations, the original POVMs have a Gaussian joint observable as well. This completes the proof.

\section{Proof of the Gaussian state-channel duality (Lemma \ref{l:gau_cj})}\label{Aps:gauss-dual}

The difference to the general case (considered above) is that in order to preserve Gaussianity, we need to do the diagonalisation of the reference state $\sigma$ ``symplectically'' (see e.g. \cite{AdessoPhD}): Let ${\bf V}_\sigma$ be the CM of $\sigma$ and ${\bf r}_\sigma$ the displacement. By Williamson's theorem \cite{Wi36} there is a symplectic matrix ${\bf S}$ such that ${\bf V}_\sigma = {\bf S}^T{\bf D}{\bf S}$ with ${\bf D}=\oplus_{k=1}^N \nu_k \id_2$, where $\nu_k$ are the symplectic eigenvalues of ${\bf V}_\sigma$, and we assume $\nu_i>1$ (full symplectic rank). This is not restrictive as any $\nu_i=1$ corresponds to a vacuum mode, which we may factor out from the system. Then $U=D_{{\bf r}_\sigma}U_{\bf S}$ diagonalises $\sigma$ in the photon number basis $|{\bf n}\rangle=|n_1,\ldots,n_N\rangle$:
\begin{equation}
U^*\sigma U= \sum_{\bf n} p_{\bf n} |{\bf n}\rangle\langle {\bf n}|, 
\quad p_{\bf n}=\prod_{k=1}^N \frac{2}{1+\nu_k} \left(\frac{\nu_k-1}{\nu_k+1}\right)^{n_k}.
\end{equation}
 Moreover, the purification $\sum_{\bf n} \sqrt{p_{\bf n}} |{\bf n}\rangle\otimes |{\bf n}\rangle$ has the CM 
 \begin{equation}
 \left( \begin{array}{cc}
{\bf D} & {\bf Z} \\
{\bf Z} & {\bf D}
\end{array}\right) \quad \text{with } {\bf Z} = \bigoplus_{i=1}^N\sqrt{\nu_i^2-1}\, \sigma_z.
\end{equation}

The eigenbasis $\{ U|{\bf n}\rangle\}$ of $\sigma$ is the one we use to construct the steering channels following the general scheme (see Lemma~1). Hence we form the purification
\begin{equation}
\Omega_\sigma = \sum_{\bf n} \sqrt{p_{\bf n}} U|{\bf n}\rangle\otimes U|{\bf n}\rangle
\end{equation}
which by \eqref{channelcov} has displacement vector ${\bf r}_\sigma\oplus {\bf r}_\sigma$ and CM
\begin{equation}
{\bf V}_{\Omega_\sigma} = ({\bf S}^T\oplus {\bf S}^T)\left( \begin{array}{cc}
{\bf D} & {\bf Z} \\
{\bf Z} & {\bf D}
\end{array}\right)({\bf S}\oplus {\bf S})= \begin{pmatrix}
{\bf V}_\sigma & {\bf S}^T{\bf Z}{\bf S} \\
{\bf S}^T{\bf Z}{\bf S} & {\bf V}_\sigma
\end{pmatrix}
\end{equation} as stated in the Lemma. Again by \eqref{channelcov}, the application of a Gaussian channel $\Lambda$ with matrices $({\bf M},{\bf N},{\bf c})$ yields the state $\rho:=(\Lambda\otimes {\rm Id})(|\Omega_\sigma\rangle\langle\Omega_\sigma|)$ with CM
\begin{align}
{\bf V} &= ({\bf M}^T\oplus {\bf I})  \begin{pmatrix}
{\bf V}_\sigma & {\bf S}^T{\bf Z}{\bf S} \\
{\bf S}^T{\bf Z}{\bf S} & {\bf V}_\sigma
\end{pmatrix}({\bf M}\oplus {\bf I})+{\bf N}\oplus {\bf 0}\nonumber\\
&= \begin{pmatrix}
{\bf M}^T{\bf V}_\sigma{\bf M}+{\bf N} & {\bf M}^T{\bf S}^T{\bf Z}{\bf S} \\
{\bf S}^T{\bf Z}{\bf SM} & {\bf V}_\sigma \label{CMrho}
\end{pmatrix},
\end{align}
and displacement ${\bf r}=({\bf M}^T{\bf r}_\sigma +{\bf c})\oplus {\bf r}_\sigma$.
Now ${\bf V}_\rho+i{\bf \Omega}\geq 0$ if and only if ${\bf C}\geq 0$ where ${\bf C}$ is the Schur complement of the block ${\bf V}_\sigma+i{\bf \Omega}_B$ in ${\bf V}_\rho+i{\bf \Omega}$. But
\begin{align}
\nonumber{\bf C} &= {\bf M}^T{\bf V}_\sigma{\bf M} +{\bf N}+ i{\bf \Omega}_A\\
&-{\bf M}^T{\bf S}^T{\bf Z}{\bf S} ({\bf V}_\sigma+i{\bf \Omega}_B)^{-1} {\bf S}^T{\bf Z}{\bf SM}\nonumber\\
&={\bf N}+i{\bf \Omega}_A+ {\bf M}^T{\bf S}^T({\bf D} - {\bf Z} ({\bf D}+i{\bf \Omega}_B)^{-1} {\bf Z}){\bf S}{\bf M}\nonumber\\
&= {\bf C}_{\bf M,N}+i{\bf \Omega}_A,\label{schur}
\end{align}
where we have used ${\bf D} - {\bf Z} ({\bf D}+i{\bf \Omega})^{-1} {\bf Z}={\bf \Omega}$
which is straightforward to verify. This shows that ${\bf C}_{\bf M,N}+i{\bf \Omega}_A\geq 0$ is \emph{equivalent} to ${\bf V}_\rho$ being a valid CM. Now for any given Gaussian state $\rho$ with CM and displacement vector
\begin{equation}\label{genrho}{
{\bf V} = \left( \begin{array}{cc}
{\bf V}_A & {\bf \Gamma}^T \\
{\bf \Gamma} & {\bf V}_\sigma
\end{array}\right), \quad {\bf r}={\bf r}_A\oplus {\bf r}_\sigma,
}\end{equation}
we can define 
\begin{equation}({\bf M}, {\bf N}, {\bf c})= (({\bf S}^T{\bf Z}{\bf S})^{-1}{\bf \Gamma},\ {\bf V}_A-{\bf M}^T{\bf V}_\sigma {\bf M}, \ {\bf r}_A-{\bf M}^T{\bf r}_\sigma),
\end{equation}
which then satisfies \eqref{CMrho}, so that ${\bf C}_{\bf M,N}+i{\bf \Omega}_A\geq 0$ by the above equivalence, showing that $({\bf M}, {\bf N}, {\bf c})$ determines a Gaussian channel $\Lambda$ with $\rho=(\Lambda\otimes {\rm Id})(|\Omega_\sigma\rangle\langle\Omega_\sigma|)$. This completes the proof.

\section{The derivation of the LHS from the joint Gaussian measurement}\label{App:LHS}

Here we show that our joint Gaussian POVM (discussed in the main text) is consistent with the LHS of \cite{WiJoDo07}. According to the general discussion in Section \ref{subsec:prel_hidden}, joint POVM $G_{\boldsymbol\lambda}$ and the LHS $\sigma_{\boldsymbol\lambda}$ are related by
$\sigma_{\boldsymbol \lambda}=\sigma^{\frac 12}G_{\boldsymbol \lambda} \sigma^{\frac 12}= {\rm tr}_A[G_{\boldsymbol\lambda} \otimes \id |\Omega_\sigma\rangle\langle \Omega_\sigma|]$. Now $\sigma_{\boldsymbol \lambda}$ has finite trace, and $\tilde\sigma_{\boldsymbol\lambda} := \sigma_{\boldsymbol\lambda}/{\rm tr}[\sigma_{\boldsymbol\lambda}]$ is an actual state; we show that it is Gaussian by computing the characteristic function $\widehat{\tilde{\sigma}_{\boldsymbol \lambda}}({\bf x}) := {\rm tr}[W({\bf x})\tilde \sigma_{\boldsymbol\lambda}]=f_{\bf x}({\boldsymbol \lambda})/f_{{\bf 0}}({\boldsymbol \lambda})$, where $f_{\bf x}(\boldsymbol \lambda):= {\rm tr}[W({\bf x})\sigma_{\boldsymbol\lambda}] = {\rm tr}[G_{\boldsymbol\lambda} \otimes W({\bf x})|\Omega_\sigma\rangle \langle \Omega_\sigma|]$. For simplicity we assume ${\bf c}=0$. Due to \eqref{eqn:gaussian_obs}, the function $f_{\bf x}$ is determined via its Fourier transform, in terms of the channel parameters $({\bf M,N,c})$. For simplicity, we will assume ${\bf c}=0$, and compute
\begin{align}
\nonumber\widehat{f}_{\bf x}({\bf p})
\nonumber &= \int e^{i{\bf p}^T{\boldsymbol\lambda}} \,{\rm tr}[G_{\boldsymbol \lambda} \otimes W({\bf x})|\Omega_\sigma\rangle \langle \Omega_\sigma|]\, d{\boldsymbol \lambda}\\
\nonumber &= {\rm tr}[\hat G({\bf p}) \otimes W({\bf x})|\Omega_\sigma\rangle \langle \Omega_\sigma|]\\
&= {\rm tr}[W({\bf Mp}) \otimes W({\bf x})|\Omega_\sigma\rangle \langle \Omega_\sigma|]e^{-\frac 14 {\bf p}^T{\bf Np}}.
\end{align}
Now by definition, the first factor in the last expression is the characteristic function of the state $\Omega_\sigma$, evaluated at ${\bf Mp}\oplus {\bf x}$; hence by \eqref{covdef} and \eqref{CMrho} we get
\begin{align}
\nonumber
\widehat{f}_{\bf x}({\bf p}) &= e^{-\frac 14 (({\bf Mp})^T\oplus {\bf x}^T){\bf V}_{\Omega_\sigma}({\bf Mp}\oplus {\bf x}) }e^{-\frac 14 {\bf p}^T{\bf Np}}\\
\nonumber &= e^{-\frac 14 ({\bf p}^T\oplus {\bf x}^T){\bf V}({\bf p}\oplus {\bf x}) }= e^{-\frac 14({\bf p}^T{\bf V}_A{\bf p} +2 {\bf p}^T{\bf \Gamma}^T{\bf x} + {\bf x}^T{\bf V}_\sigma{\bf x})}\\
&= e^{-\frac 14({\bf p-\boldsymbol \mu_{\bf x}})^T{\bf V}_A({\bf p-\boldsymbol\mu_{\bf x}})}\,\,e^{ - \frac  14 {\bf x}^T({\bf V}_\sigma- \boldsymbol \Gamma {\bf V}_A^{-1}\boldsymbol \Gamma^T){\bf x}}
\end{align}
where $\boldsymbol \mu_{\bf x} = -{\bf V}_A^{-1}{\bf \Gamma}^T{\bf x}$, and we have used the notation \eqref{genrho}. Taking the inverse Fourier transform we obtain
\begin{equation}
f_{\bf x}({\boldsymbol \lambda}) = C e^{-{\boldsymbol\lambda}^T{\bf V}_A^{-1}{\boldsymbol\lambda} -i{\boldsymbol\lambda}^T{\boldsymbol \mu}_{\bf x}}\,\,e^{ - \frac  14 {\bf x}^T({\bf V}_\sigma- \boldsymbol \Gamma {\bf V}_A^{-1}\boldsymbol \Gamma^T){\bf x}}
\end{equation}
where $C$ depends only on ${\bf V}_A$. Hence
$\widehat{\tilde{\sigma}_{\boldsymbol \lambda}}({\bf x}) = f_{\bf x}({\boldsymbol \lambda})/f_{{\bf 0}}({\boldsymbol \lambda})=e^{ - \frac  14 {\bf x}^T({\bf V}_\sigma- \boldsymbol \Gamma {\bf V}_A^{-1}\boldsymbol \Gamma^T){\bf x}  +i({\bf \Gamma} {\bf V}^{-1}_A{\boldsymbol\lambda})^T{\bf x}}$, so by \eqref{covdef}, $\tilde \sigma_{\boldsymbol \lambda}$ is Gaussian with CM and displacement
\begin{align}
{\bf V}_{\boldsymbol \lambda} &= {\bf V}_\sigma- \boldsymbol \Gamma {\bf V}_A^{-1}\boldsymbol \Gamma^T, & {\bf r}_{\boldsymbol \lambda } &= -{\bf \Gamma}{\bf V}_A^{-1}{\boldsymbol\lambda}.
\end{align}
Furthermore, each $\tilde \sigma_\lambda$ occurs in the LHS decomposition with Gaussian probability $p_{\boldsymbol \lambda} = {\rm tr}[\sigma_{\boldsymbol \lambda}]=f_{\bf 0}({\boldsymbol \lambda}) \propto e^{-{\boldsymbol\lambda}^T{\bf V}_A^{-1}{\boldsymbol \lambda}}$. By changing the hidden variable $\boldsymbol \lambda$ to ${\bf r}_{\boldsymbol \lambda}$ we recover exactly the LHS of \cite{WiJoDo07}.



%

\end{document}